\documentclass[a4paper,draft]{amsproc}
\usepackage{amssymb,amsthm,amsmath}
\usepackage[hyphens]{url} 
\usepackage{breqn}
\usepackage[dvips]{graphicx}

\usepackage{cmdtrack}

\theoremstyle{plain}
 \newtheorem{thm}{Theorem}[section]
 \newtheorem{prop}{Proposition}[section]
 \newtheorem{remark}{Remark}[section]

\theoremstyle{definition}
 \newtheorem{exm}{Example}[section]
 \newtheorem{dfn}{Definition}[section]
\theoremstyle{remark}
 
 \numberwithin{equation}{section}
 \newcommand{\be}{\begin{equation}}
\newcommand{\ee}{\end{equation}}

\renewcommand{\leq}{\leqslant}
\renewcommand{\geq}{\geqslant}

\title[On some integral means]{On Some Integral means}



\author[Khoshnasib-Zeinabad]{\bfseries Fariba Khoshnasib-Zeinabad} 

\address{ 
Department of Mathematics \\ 
Earlham College \\ 
Richmond, IN\\
United States}
\email{khoshfa@earlham.edu}
\author[Mehrabi]{\bfseries Mohammadhossein Mehrabi}

\address{ 
Department of Mathematics \\ 
Iran University of Science and Technology  \\ 
Province of Tehran \\
Iran}
\email{mhmehrabi@gmail.com}

\thanks{Communicated by ...} 

\begin{document}

{\begin{flushleft}\baselineskip9pt\scriptsize

\end{flushleft}}
\vspace{18mm} \setcounter{page}{1} \thispagestyle{empty}

\begin{abstract}
 \sloppy Harmonic, Geometric, Arithmetic, Heronian and Contra-harmonic means have been studied by many mathematicians. In 2003, H. Eves studied these means from geometrical point of view and established some of the inequalities between them in using a circle and its radius. In 1961, E. Beckenback and R. Bellman introduced several inequalities corresponding to means. In this paper, we will introduce the concept of mean functions and integral means and give bounds on some of these mean functions and integral means.
\end{abstract}
\maketitle

\section{Introduction}
In their book of inequalities, Beckenback and  Bellman established several inequalities between arithmetic, harmonic and contra-harmonic means\cite{mbook}. These means are defined in the following paragraph, based on the original text by Eve\cite{main}.\\
Let $a,b>0$ and $a\neq b.$ Putting together the results from works of several mathematicians, in particular Taneja established that$\
\max\{a,b\}>C>r>g>A>Hn>G>H>\min\{a,b\}$ in \cite{7means} and \cite{7means2}, where
$C=\frac{a^2+b^2}{a+b}$ is contraharmonic mean,$\
r=\sqrt{\frac{a^2+b^2}{2}}$ is root square mean,
$g=\frac{2(a^2+ab+b^2)}{3(a+b)}$ is gravitational mean (also called centroidal mean),
$A=\frac{a+b}{2}$ is arithmetic mean, $Hn=\frac{a+\sqrt{ab}+b}{3}$
is Heronian mean, $G=\sqrt{ab}$ is geometric mean and
$H=\frac{2ab}{a+b}$ is harmonic mean of $a$ and $b$.\\In this
paper we introduce the notion of a mean function and utilize it to define some integral means of $a$ and $b$ and then we establish 
some inequalities corresponding to those mean functions and integral means.
\section{Definitions and Main Theorems}
All the means that appear in this paper are functions F with conditions a and b satisified:\\
a) $F:\mathbb{R}_{+}^2\to\mathbb{R}_{+},$\ \ where $\min\{x,y\}\leq
F(x,y)\leq\max\{x,y\},$Provided that $(x,y)\in\mathbb{R}_{+}^2$,\\ b) $F(x,y)=F(y,x),$ such that $(x,y)\in\mathbb{R}_{+}^2.$ Consequently,  $F(x,x)=x$ where  $x\in\mathbb{R}_{+}$.\\ We say $F$ is a
mean function when the two above conditions are satisfied.\\All throughout the paper we are assuming that $a,b>0$ and without loss of generality can assume
 $b\geq a$, by symmetry.\\\\
 \begin{dfn} \it Let $M$ be a mean of $a$ and
 $b$. We define $M_A:=2A-M.$ to be $A$-complementary (arithmetic complementary) of $M$. \end{dfn}
 
 It is obvious that $M_A$ and $M_G$ are means of $a$ and $b$.
 
 \begin{thm}
 Let $M\in\mathcal{R}(\mathbb{R}_{+}^2)$ be a mean function.
Then\\
(i)\[\mathcal{I}_M:=\mathcal{I}_M(a,b):=\begin{cases}
    \frac{1}{(b-a)^2}\int_a^b\int_a^b M(x,y)\ dxdy, &\text{if $a\neq b$}, \\
    a, & \text{if $a=b$} \
\end{cases}\]  \\
is a mean of $a$ and $b$,\\\\ (ii) $\mathcal{J}_M:=\mathcal{J}_M(a,b):=3\mathcal{I}_M(a,b)-2A(a,b)$ is a mean of $a$ and  $b$ and finally \\\\(iii) $\frac{2}{3}A<\mathcal{I}_M<\frac{4}{3}A.$
\end{thm}

\begin{proof} Let
$b>a$.\\(i):\[\min\{a,b\}=a<\frac{2a+b}{3}=\frac{1}{(b-a)^2}\int_a^b\int_a^xy\
dydx+\frac{1}{(b-a)^2}\int_a^b\int_x^bx\
dydx=\]\[\frac{1}{(b-a)^2}\int_a^b\int_a^b\min\{x,y\}\
dxdy\leq\mathcal{I}_M(a,b)\leq\frac{1}{(b-a)^2}\int_a^b\int_a^b\max\{x,y\}\
dxdy\]\[=\frac{1}{(b-a)^2}\int_a^b\int_a^xx\
dydx+\frac{1}{(b-a)^2}\int_a^b\int_x^by\
dydx=\frac{a+2b}{3}<b=\max\{a,b\}.\]Also, it is obvious that$\
\mathcal{I}_M\ $is symmetric.\\(ii):\\By proof of (i), we
have$\frac{2a+b}{3}\leq\mathcal{I}_M(a,b)\leq\frac{a+2b}{3}.\label{2}
So,$\[a\leq3\mathcal{I}_M(a,b)-(a+b)\leq b.\]$\\(iii):
$\[\frac{2}{3}<\frac{2(2a+b)}{3(a+b)}\leq\frac{\mathcal{I}_M(a,b)}{A(a,b)}\leq\frac{2(a+2b)}{3(a+b)}<\frac{4}{3},\]
multiplying by $A(a,b)$ we get the result.
\end{proof}
\begin{prop}
Let $M,M_1
$and$ M_2$ be mean functions and $\lambda\in\mathbb{R}$.
Then\\\\(i)
$M_1>M_2\Rightarrow\mathcal{I}_{M_{1}}(a,b)>\mathcal{I}_{M_{2}}(a,b)\
\ $and$\ \ \mathcal{J}_{M_{1}}(a,b)>\mathcal{J}_{M_{2}}(a,b),\ \
a\neq b,$\\\\(ii) $\mathcal{I}_{\lambda M_{1}+(1-\lambda)
M_{2}}=\lambda\mathcal{I}_{M_{1}}+(1-\lambda)\mathcal{I}_{M_{2}},\
$if$\ \lambda M_{1}+(1-\lambda) M_{2}\ $is a mean
function.\\\\In particular, $\mathcal{I}_{M_{A}}=(\mathcal{I}_M)_{A}\
$ and$\ \
\mathcal{I}_{\mathcal{J}_{M}}=\mathcal{J}_{\mathcal{I}_{M}}$.
\end{prop}
\begin{proof}
Proof is easily done by straightforward calculations.
\end{proof}
Here are some examples where the above proposition is used:\\
Let $ a\neq b.$\\
\begin{exm} 
\[\mathcal{I}_A(a,b)=\frac{1}{(b-a)^2}\int_a^b\int_a^b\frac{1}{2}(x+y)\ dxdy=\frac{1}{2(b-a)^2}\int_a^b\left[\frac{x^2}{2}+xy\right]_a^b dy\]
\[=\frac{1}{4(b-a)^2}\left[y(b^2-a^2)+y^2(b-a)\right]_a^b=\frac{b+a}{2}=A(a,b).\]\end{exm}
\begin{exm}\[\mathcal{I}_G(a,b)=\frac{1}{(b-a)^2}\int_a^b\int_a^b\sqrt{xy}\ dxdy=\bigg(\frac{1}{b-a}\int_a^b\sqrt{t}\ dt\bigg)^2 =\bigg(\frac{2(b^{\frac{3}{2}}-a^{\frac{3}{2}})}{3(b-a)}\bigg)^2\]
\[=g^2(\sqrt{a},\sqrt{b}).\]\end{exm}
\begin{exm}\[\mathcal{I}_H(a,b)=\frac{2}{(b-a)^2}\int_a^b\int_a^b\frac{xy}{x+y}\ dxdy=\frac{2}{(b-a)^2}\int_a^b\left[xy-y^2\ln(x+y)\right]_a^b dy\]
\[=\frac{2}{3(b-a)^2}\left[y^2(b-a)+y(b^2-a^2)-(y^3+b^3)\ln(y+b)+(y^3+a^3)\ln(y+a)\right]_a^b=\]
\[\frac{4}{3}\bigg(2A(a,b)+\frac{1}{(b-a)^2}\bigg(a^3\ln\frac{A(a,b)}{a}+b^3\ln\frac{A(a,b)}{b}\bigg)\bigg),\ \ \ \ a\neq b.\]\end{exm}
\begin{exm}Let$ a<b$ , then:\\ \[\sqrt{2}(b-a)^2\mathcal{I}_{r}(a,b)=\int_a^b\int_a^b\sqrt{x^2+y^2}\ dxdy\]
\[=\int_{\tan^{-1}\tfrac{a}{b}}^{\tfrac{\pi}{4}}\int_{\tfrac{a}{\sin\theta}}^{\tfrac{b}{\cos\theta}}\rho^2\
d\rho d\theta\ +
\int_{\tfrac{\pi}{4}}^{\tan^{-1}\tfrac{b}{a}}\int_{\tfrac{a}{\cos\theta}}^{\tfrac{b}{\sin\theta}}\rho^2\
d\rho
d\theta.\]\\
Double integrals in the above expression are easily calculated and the final result is:
\[
\mathcal{I}_{r}(a,b)=\frac{1}{3\sqrt{2}(b-a)^2}\bigg(\big(\sqrt{2}+\ln(1+\sqrt{2})\big)(a^3+b^3)-a^3\ln\big(\frac{b+\sqrt{a^2+b^2}}{a}\big)-...\]\[
b^3\ln\big(\frac{a+\sqrt{a^2+b^2}}{b}\big)-2ab\sqrt{a^2+b^2}\bigg),\
\ \ \ a\neq b.\]\end{exm} Similarly, by Proposition 2.1 (ii), we will
have:\\
\begin{exm}Let $a\neq b
$, then\\
\[ \mathcal{I}_C(a,b)=(\mathcal{I}_H)_{A}(a,b)=-\frac{2}{3}\bigg(A(a,b)+\frac{2}{(b-a)^2}\bigg(a^3\ln\frac{A(a,b)}{a}+b^3\ln\frac{A(a,b)}{b}\bigg)\bigg)\]\end{exm}
\begin{exm}Let $a\neq b$, then\\
\[ \mathcal{I}_{g}(a,b)=\frac{4}{3}\mathcal{I}_A(a,b)-\frac{1}{3}\mathcal{I}_H(a,b)=\]
\[\frac{4}{9}\bigg(A(a,b)-\frac{1}{(b-a)^2}\bigg(a^3\ln\frac{A(a,b)}{a}+b^3\ln\frac{A(a,b)}{b}\bigg)\bigg).\]\end{exm}
\begin{exm}\[\hspace{2mm}\mathcal{I}_{Hn}(a,b)=\frac{2}{3}\mathcal{I}_A(a,b)+\frac{1}{3}\mathcal{I}_G(a,b)=\frac{1}{3}\big(2A(a,b)+g^2(\sqrt{a},\sqrt{b})\big)
\]\[=\frac{2(13A^2+13AG+G^2)}{27(A+G)}\]\end{exm}
and finally,

\begin{exm}\[\mathcal{I}_{\frac{A+G}{2}}(a,b)=\frac{1}{2}\mathcal{I}_A(a,b)+\frac{1}{2}\mathcal{I}_G(a,b)=\frac{1}{2}\big(A(a,b)+g^2(\sqrt{a},\sqrt{b})\big)\]\[
=\frac{17A^2+17AG+2G^2}{18(A+G)}.\] \\
Because$
\mathcal{I}_A=\mathcal{J}_A=A,$
\end{exm}
By Proposition 2.1 (i), we can propose the following:\\\
\begin{prop} Let$\ M\ $be a mean function
and$\ a\neq b$, then\\\\(i)$\
M>(<)A\Rightarrow\mathcal{I}_M>(<)A,$\\\\(ii)$\
M>(<)A\Rightarrow\mathcal{J}_M>(<)A,$\\\\(iii)$\ \
\mathcal{I}_C>\mathcal{I}_r>\mathcal{I}_g>A>\mathcal{I}_{Hn}>\mathcal{I}_G>\mathcal{I}_H\
\ $and$\ \
\mathcal{J}_C>\mathcal{J}_r>\mathcal{J}_g>A>\mathcal{J}_{Hn}>\mathcal{J}_G>\mathcal{J}_H$.\end{prop}
 By proposition 2.1 (ii) and proposition 2.2(iii), we infer\[\ \mathcal{I}_r(a,b)>\frac{4}{3}A(a,b)-\frac{1}{3}\mathcal{I}_H(a,b)>A(a,b),
\ \ \ \ a\neq
b.\]Therefore,\[3\mathcal{I}_r(a,b)-2A(a,b)>2A(a,b)-\mathcal{I}_H(a,b)>A(a,b),\
\ \ \ a\neq b.\]In other words,\[\mathcal{J}_r(a,b)>\mathcal{I}_C(a,b)>A(a,b),\ \ \ \ a\neq
b.\]Similarly,\[A(a,b)>\mathcal{I}_G(a,b)>\mathcal{J}_G(a,b),\ \ \
\ a\neq b,\]\\\[\mathcal{I}_C(a,b)+\mathcal{I}_G(a,b)>2A(a,b)>\mathcal{I}_r(a,b)+\mathcal{I}_H(a,b),\
\ \ \ a\neq
b\]and\[\mathcal{J}_C(a,b)+\mathcal{J}_G(a,b)>2A(a,b)>\mathcal{J}_r(a,b)+\mathcal{J}_H(a,b),\
\ \ \ a\neq b.\]

\begin{prop}

 If $a\neq b$,
then\\\\(i) $\ \frac{8}{9}A<\mathcal{I}_G<A$,\\\\(ii)$\
\frac{8(1-\ln2)}{3}A<\mathcal{I}_H<A$,\\\\(iii) $\
A<\mathcal{I}_C<\frac{2(-1+4\ln2)}{3}A$,\\\\(iv) $\
\frac{26}{27}A<\mathcal{I}_{Hn}<A$,\\\\(v) $\
A<\mathcal{I}_g<\frac{4(1+2\ln2)}{9}A$,\\\\(vi) $\
A<\mathcal{I}_r<\frac{2+\sqrt{2}\ln(1+\sqrt{2})}{3}A,$\\\\(vii) $\
\frac{2}{3}A<\mathcal{J}_G<A$,\\\\(viii)$\
2(3-4\ln2)A<\mathcal{J}_H<A$,\\\\(ix) $\
A<\mathcal{J}_C<4(-1+2\ln2)A$,\\\\(x) $\
\frac{8}{9}A<\mathcal{J}_{Hn}<A$,\\\\(xi) $\
A<\mathcal{J}_g<\frac{2(-1+4\ln2)}{3}A$,\\\\(xii) $\
A<\mathcal{J}_r<\sqrt{2}\big(\ln(1+\sqrt{2})\big)A,$\\\\where$\ \
1,\ \ \frac{8}{9},\ \ \frac{8(1-\ln2)}{3},\ \
\frac{2(-1+4\ln2)}{3},\ \ \frac{26}{27},\ \ \frac{4(1+2\ln2)}{9},\
\ \frac{2+\sqrt{2}\ln(1+\sqrt{2})}{3},\ \ \frac{2}{3},\ \
2(3-4\ln2),$\\$4(-1+2\ln2),\ \ $and$\ \ \sqrt{2}\ln(1+\sqrt{2})\ \
$are the best possible bounds we found for the inequalities between the integral means and the mean functions.\end{prop}

\begin{proof}
(i):\ If we
take $b=at^2,\ \ t>1,\ $\ then the following will be concluded:\\$\ \
f_1(t):=\frac{8(t^2+t+1)^2}{9(t^2+1)(t+1)^2}=\frac{\mathcal{I}_G(a,b)}{A(a,b)}.$\\Taking the derivative, we get:$f_1^{'}(t)=\frac{-16t(t^3-1)}{9(t^2+1)^2(t+1)^3}<0.\
\ $Therefore,$\ \ f_1\ \ $is strictly decreasing. So,$\ \
\lim_{t\to\infty}f_1(t)=\frac{8}{9}<f_1(t)<\lim_{t\to1^{+}}f_1(t)=1,\
\ t>1.$\\(ii):\ If we take $b=at,\ \ t>1,\ $\ then we will have$\
\ f_2(t):=\frac{8}{3}\big(1+\frac{(t^3+1)\ln\frac{t+1}{2}-t^3\ln
t}{(t-1)^2(t+1)}\big)=\frac{\mathcal{I}_H(a,b)}{A(a,b)}.$\\$f_2^{'}(t)=\frac{-8f_3(t)}{3(t^3-t^2-t+1)^2},\
\ $where$\ \ f_3(t):=(t+1)^3\ln\frac{t+1}{2}-(t^3+3t^2)\ln
t+t^3-t^2-t+1.$\\$f_3^{'}(t)=3(t+1)^2\ln\frac{t+1}{2}-3(t^2+2t)\ln
t+3t^2-3t.\ \ f_3^{''}(t)=6(t+1)\ln\frac{t+1}{2}-6(t+1)\ln
t+6t-6.\ \ f_3^{'''}(t)=6\ln\frac{t+1}{2}-6\ln t+6-\frac{6}{t}.\ \
f_3^{''''}(t)=\frac{6}{t^2(t+1)}>0.\ $Therefore, $f_3^{'''}$ is
strictly increasing, so
$f_3^{'''}(t)>\lim_{t\to1^{+}}f_3^{'''}(t)=0,\ \ t>1.\
$Consequently, $f_3^{''}$ is strictly increasing, hence
$f_3^{''}(t)>\lim_{t\to1^{+}}f_3^{''}(t)=0,\ \ t>1.\ $Therefore,
$f_3^{'}$ is strictly increasing, so
$f_3^{'}(t)>\lim_{t\to1^{+}}f_3^{'}(t)=0,\ \ t>1.\ $Thus, $f_3$ is
strictly increasing, hence $f_3(t)>\lim_{t\to1^{+}}f_3(t)=0,\ \
t>1.\ $Consequently,$\ \ f_2^{'}(t)<0,\ \ t>1.\ $Therefore, $f_2$
is strictly decreasing, so
$\frac{8(1-\ln2)}{3}=\lim_{t\to\infty}f_2(t)<f_2(t)<\lim_{t\to1^{+}}f_2(t)=1,\
\ t>1.$\\(iii):\ $\ a\neq
b\Rightarrow\frac{8(1-\ln2)}{3}<\frac{\mathcal{I}_H(a,b)}{A(a,b)}<1\Rightarrow$\\$\frac{2(-1+4\ln2)}{3}=2-\frac{8(1-\ln2)}{3}>2-\frac{\mathcal{I}_H(a,b)}{A(a,b)}=
\frac{\mathcal{I}_C(a,b)}{A(a,b)}>2-1=1.$\\(iv):\ $\ a\neq
b\Rightarrow\frac{8}{9}A(a,b)<\mathcal{I}_G(a,b)<A(a,b)\Rightarrow\frac{26A}{27}=\frac{2A}{3}+\frac{8A}{27}<\frac{2A}{3}+\frac{1}{3}\mathcal{I}_G(a,b)=\mathcal{I}_{Hn}(a,b)
<\frac{2A}{3}+\frac{A}{3}=A.$\\(v):\ $\ a\neq
b\Rightarrow\frac{8(1-\ln2)}{3}<\frac{\mathcal{I}_H(a,b)}{A(a,b)}<1\Rightarrow$\\$\frac{4(1+2\ln2)}{9}=\frac{4}{3}-\frac{8(1-\ln2)}{9}>\frac{4}{3}-\frac{\mathcal{I}_H(a,b)}{3A(a,b)}=
\frac{\mathcal{I}_g(a,b)}{A(a,b)}>\frac{4}{3}-\frac{1}{3}=1.$\\(vi):\
$$f_4(t):=\frac{\sqrt{2}}{3}\bigg(\frac{k(t^3+1)-t^3\ln\frac{1+\sqrt{1+t^2}}{t}-\ln(t+\sqrt{1+t^2})-2t\sqrt{1+t^2}}{(t+1)(t-1)^2}\bigg)=\frac{\mathcal{I}_r(a,b)}{A(a,b)},$$
where$\ a=bt,\ \ 0<t<1\ $and$\ k:=\sqrt{2}+\ln(1+\sqrt{2}).\ $The we
have:\\$\ f_4^{'}(t)=\frac{\sqrt{2}f_5(t)}{3(t+1)^2(t-1)^3}$,where\[
f_5(t):=(3t^2+2t+3)\sqrt{1+t^2}+t^2(t+3)\ln\frac{1+\sqrt{1+t^2}}{t}+...\]\[(3t+1)\ln(t+\sqrt{1+t^2})-k(t+1)^3.\]\\
Therefore,
$$f_5^{'}(t)=(3t^2+6t)\ln\frac{1+\sqrt{1+t^2}}{t}+3\ln(t+\sqrt{1+t^2})+(9t+3)\sqrt{1+t^2}-3k(t+1)^2,$$
$$f_5^{''}(t)=(6t+6)\ln\frac{1+\sqrt{1+t^2}}{t}+\frac{(18t^2+6)}{\sqrt{1+t^2}}-6k(t+1),$$
$$f_5^{'''}(t)=6\bigg(\frac{3t^4-t^3+4t^2-t-1}{t(t^2+1)^{\frac{3}{2}}}+\ln\frac{1+\sqrt{1+t^2}}{t}-k\bigg)$$and
\[f_5^{''''}(t)=\frac{6(1-t+8t^2-t^3+t^4)}{t^2(t^2+1)^{\frac{5}{2}}}.\]Since$\ t\in(0,1),\ $so$\ 1-t+8t^2-t^3+t^4=(1-t)+(8-t+t^2)t^2>0.\ $Hence,$\
f_5^{''''}>0\ $on$\ (0,1).$\\Consequently,$\ f_5^{'''}\ $will be
strictly increasing. Therefore,$\
f_5^{'''}(t)<\lim_{t\to1^{-}}f_5^{'''}(t)=6(\sqrt{2}+\ln(1+\sqrt{2})-k)=0,\
$for$\ 0<t<1.\ $Thus,$\ f_5^{''}\ $is strictly decreasing.
Hence,$\
f_5^{''}(t)>\lim_{t\to1^{-}}f_5^{''}(t)=12\ln(1+\sqrt{2})+12\sqrt{2}-12k=0,\
$for$\ 0<t<1.\ $So,$\ f_5^{'}\ $will be strictly increasing.
Therefore,$\
f_5^{'}(t)<\lim_{t\to1^{-}}f_5^{'}(t)=12\ln(1+\sqrt{2})+12\sqrt{2}-12k=0,\
$for$\ 0<t<1.\ $Hence,$\ f_5\ $is strictly decreasing. So,$\
f_5(t)>\lim_{t\to1^{-}}f_5(t)=8\ln(1+\sqrt{2})+8\sqrt{2}-8k=0,\
$on$\ (0,1)\ $and consequently$\ f_4^{'}<0\ $on$\ (0,1).\ $Thus,$\
f_4\ $will be strictly decreasing 0n$\ (0,1).\ $Therefore,$\
1=\lim_{t\to1^{-}}<f_4(t)<f_4(t)=\frac{\mathcal{I}_r(a,b)}{A(a,b)}<\lim_{t\to0^{+}}f_4(t)=\frac{\sqrt{2}k}{3}=\frac{2+\sqrt{2}\ln(1+\sqrt{2})}{3},\
$for$\ 0<t<1.$
\\By
(i), (ii), (iii), (iv), (v) and (vi), (vii), (viii), (ix), (x),
(xi) and (xii) are straightforward.
\end{proof}
\begin{thm}
 Let $\ M\in\mathcal{R}(\mathbb{R}_{+}^2)\ $ be a mean function
and $\ \varphi:[\alpha_0,\beta_0]\to[0,\infty) \ $be an integrable
function; that is, $\ \varphi\in\mathcal{R}([\alpha_0,\beta_0])$,
where$\ \alpha_0,\beta_0\in\mathbb{R}\ $and$\ \alpha_0<\beta_0.\
$Besides,$\ \psi:\mathbb{R}_{+}\to\mathbb{R}_{+}\ $is a Lipschitz
function with the constant of 1; that
is\[|\psi(x)-\psi(y)|\leq|x-y|,\ \ \ \
(x,y)\in\mathbb{R}_{+}^2.\]Then $\mathcal{S}_{M,\varphi,\psi}:=\mathcal{S}_{M,\varphi,\psi}(a,b)$, defined in the following way:\\

    \[\mathcal{S}_{M,\varphi,\psi}:=-\bar{\varphi}+A(a,b)-A(\psi(a),\psi(b))+\frac{1}{\beta_0-\alpha_0}\int_{\alpha_0}^{\beta_0}M\big(\varphi(t)+\psi(a),\varphi(t)+\psi(b)\big)\
    dt\]
is a mean of $a$ and $b$, where\\
\[\bar{\varphi}:=\frac{1}{\beta_0-\alpha_0}\int_{\alpha_0}^{\beta_0}\varphi(t)\
dt.\]\end{thm}
\begin{proof}
\[\min\{a,b\}=A(a,b)-\frac{|a-b|}{2}\leq A(a,b)-\frac{|\psi(a)-\psi(b)|}{2}=\]\[A(a,b)-A(\psi(a),\psi(b))+\min\{\psi(a),\psi(b)\}\leq
\mathcal{S}_{M,\varphi,\psi}\leq\]\[
A(a,b)-A(\psi(a),\psi(b))+\max\{\psi(a),\psi(b)\}\]\[=A(a,b)+\frac{|\psi(a)-\psi(b)|}{2}\leq
A(a,b)+\frac{|a-b|}{2}=\max\{a,b\}.\]Also, it is obvious that$\
\mathcal{S}_{M,\varphi,\psi}\ $is symmetric.
\end{proof}

\begin{remark}
Since$\
\varphi\in\mathcal{R}([\alpha_0,\beta_0])\Leftrightarrow\varphi\circ\eta\in\mathcal{R}([0,1]),\
$\\where$\ \eta(t):=(\beta_0-\alpha_0)t+\alpha_0.\ $So, without
loss of generality, we can assume$\ \alpha_0=0\ $and\\$\
\beta_0=1.$
\end{remark}
\begin{remark} If$\ \ \psi(a)=\psi(b),\ \ $then$\ \
\mathcal{S}_{M,\varphi,\psi}(a,b)=A(a,b).$\end{remark}
\begin{remark}If$\
\varphi=c\geq0\ $($\ c$ is a constant),\ \
then\[\mathcal{S}_{M,c,\psi}=\mathcal{S}_{M,c,\psi}(a,b)=-c+M(c+\psi(a),c+\psi(b))-A(\psi(a),\psi(b))+A(a,b).\]\end{remark}
\begin{remark}\[(\mathcal{S}_{M,c,\psi})_{A}=(\mathcal{S}_{M,c,\psi})_{A}(a,b)=c-M(c+\psi(a),c+\psi(b))+A(\psi(a),\psi(b))+A(a,b).\]\end{remark}
\begin{prop}Let $M,M_1$ and $M_2$ be mean functions
and$\ \lambda\in\mathbb{R}.\ $ Then\\\\(i) $\
\mathcal{S}_{A,\varphi,\psi}=A,$\\\\(ii)
$M_1>M_2\Rightarrow\mathcal{S}_{M_{1},\varphi,\psi}>\mathcal{S}_{M_{2},\varphi,\psi},\
\ a\neq b,$\\\\In particular, $\
M>(<)A\Rightarrow\mathcal{S}_{M,\varphi,\psi}>(<)A,$\\\\(iii)$\
\mathcal{S}_{\lambda M_1+(1-\lambda)
M_2,\varphi,\psi}=\lambda\mathcal{S}_{M_1,\varphi,\psi}+(1-\lambda)\mathcal{S}_{M_2,\varphi,\psi},\
$if$\ \lambda M_{1}+(1-\lambda) M_{2}\ $is a mean
function.\\
In particular,$
\mathcal{S}_{M_{A},\varphi,\psi}=(\mathcal{S}_{M,\varphi,\psi})_A.$\end{prop}
\begin{proof}
It is straightforward by direct calculations.\end{proof}
\begin{exm}
Let$\ \ M=G,\ \ \varphi(t)=c\geq0\ $($\
c$ is a constant) and $\check{\psi}(t)=\frac{t-\sin t}{2},\
$then\[\mathcal{N}_c:=\mathcal{N}_c(a,b):=\mathcal{S}_{G,c,\check{\psi}}(a,b)=A(a,b)-\frac{1}{4}\bigg(\sqrt{2c+a-\sin
a}-\sqrt{2c+b-\sin
b}\bigg)^2.\]In particular,\[\mathcal{N}_0=\mathcal{N}_0(a,b)=A(a,b)-\frac{1}{4}\bigg(\sqrt{a-\sin
a}-\sqrt{b-\sin b}\bigg)^2.\]We can
see\[(\mathcal{N}_c)_{A}=(\mathcal{N}_c)_{A}(a,b)=A(a,b)+\frac{1}{4}\bigg(\sqrt{2c+a-\sin
a}-\sqrt{2c+b-\sin
b}\bigg)^2.\]In particular,\[(\mathcal{N}_0)_{A}=(\mathcal{N}_0)_{A}(a,b)=A(a,b)+\frac{1}{4}\bigg(\sqrt{a-\sin
a}-\sqrt{b-\sin b}\bigg)^2.\]\end{exm}
\begin{exm}
Let$\ \ M=G,\ \
\varphi(t)=c\geq0\ $($\ c$ is a constant) and
$\tilde{\psi}(t)=\ln(t^2+1),\
$then\[L_c:=L_c(a,b):=\mathcal{S}_{G,c,\tilde{\psi}}(a,b)=-c+A(a,b)-\ln(\sqrt{(a^2+1)(b^2+1)})+...\]\[\sqrt{(c+\ln(a^2+1))(c+\ln(b^2+1))}
.\]\\In particular,\[L_0=L_0(a,b)=A(a,b)-\ln\big(\sqrt{(a^2+1)(b^2+1)}\big)+\sqrt{(\ln(a^2+1))(\ln(b^2+1))}.\]
Also we can see\[(L_c)_{A}=(L_c)_{A}(a,b)=c+A(a,b)+\ln(\sqrt{(a^2+1)(b^2+1)})-...\]\[\sqrt{(c+\ln(a^2+1))(c+\ln(b^2+1))}.\]\\
In particular,\[
(L_0)_{A}=(L_0)_{A}(a,b)=A(a,b)+\ln(\sqrt{(a^2+1)(b^2+1)})-\sqrt{(\ln(a^2+1))(\ln(b^2+1))}.\]\end{exm}
\begin{exm} Let$\ \ M=H,\ \ \varphi(t)=Id(t)=t,\
$then\[J_{\psi}:=J_{\psi}(a,b):=\mathcal{S}_{H,Id,\psi}(a,b)=A(a,b)-\bigg(\frac{\psi(a)-\psi(b)}{2}\bigg)^2
\ln\bigg(1+\frac{1}{A(\psi(a),\psi(b))}\bigg).\]In particular,\begin{equation}
   J_{Id}=J_{Id}(a,b)=A(a,b)-\bigg(\frac{a-b}{2}\bigg)^2
\ln\bigg(1+\frac{1}{A(a,b)}\bigg). 
\end{equation}
We can
see\[(J_{\psi})_{A}=(J_{\psi})_{A}(a,b)=A(a,b)+\bigg(\frac{\psi(a)-\psi(b)}{2}\bigg)^2
\ln\bigg(1+\frac{1}{A(\psi(a),\psi(b))}\bigg).\]In particular,\[(J_{Id})_{A}=(J_{Id})_{A}(a,b)=A(a,b)+\bigg(\frac{a-b}{2}\bigg)^2
\ln\bigg(1+\frac{1}{A(a,b)}\bigg).\]
\end{exm}
\begin{exm} Let$ M=G\ $
and $\ \psi(t)=Id(t)=t,\
$then\begin{equation}
    I_{\varphi}:=I_{\varphi}(a,b):=\mathcal{S}_{G,\varphi,Id}(a,b)=-\bar{\varphi}+\int_0^1\sqrt{(\varphi(t)+a)(\varphi(t)+b)}dt.
\end{equation}In particular,\begin{equation}
    I_{Id}=I_{Id}(a,b)=
-\frac{1}{2}+F_1(a+1,b+1)-F_1(a,b)-F_2(a+1,b+1)+F_2(a,b),
\end{equation} 
where\begin{equation}
    F_1(x,y):={\frac{1}{4}}(x+y)\sqrt{xy},\hskip2cm
F_2(x,y):={\frac{1}{4}}(x-y)^2\ln\frac{\sqrt{x}+\sqrt{y}}{\sqrt{2}}.
\end{equation}

We can
see\[(I_{\varphi})_{A}=(I_{\varphi})_{A}(a,b)=a+b+\bar{\varphi}-\int_0^1\sqrt{(\varphi(t)+a)(\varphi(t)+b)}\
dt.\]In particular,\[(I_{Id})_{A}=(I_{Id})_{A}(a,b)=
a+b+\frac{1}{2}-F_1(a+1,b+1)+F_1(a,b)+F_2(a+1,b+1)-F_2(a,b).\]If\
$a\neq b,\ $by proposition4 (ii), we will
have\[\mathcal{N}_c<A,\hskip2cm L_c<A,\hskip2cm
J_{\psi}<A,\hskip2cm I_{\varphi}<A.\]
\end{exm}
\begin{prop}
If$\ \ a\neq b,\ \ $then$\ \ I_{\varphi}>G,\ \ $for every$\
\varphi,\ $whose support is positive measure; equivalently, there
exists $\ S\subseteq[0,1]\ \ $with$\ \ |S|>0,\ \ $such that $\
\varphi(t)>0\ \ $for every $\ t\in S.$\end{prop}
\begin{proof} Let$\
a\neq
b.$We know that \[I_{\varphi}>G\]is equivalent to \[\int_0^1\bigg(\sqrt{(\varphi(t)+a)(\varphi(t)+b)}-\varphi(t)-G(a,b)\bigg)\
dt>0.\] Let us start our argument by working with the integrand:
\[\sqrt{(\varphi(t)+a)(\varphi(t)+b)}-\varphi(t)-G(a,b)>(\geq)0\] This results in: \[(\sqrt{a}-\sqrt{b})^2\varphi(t)>(\geq)0.\]Thus,
$
    \sqrt{(\varphi(t)+a)(\varphi(t)+b)}-\varphi(t)-G(a,b)\geq0,$
$t\in[0,1]$\label{3}
and
\\
$ \sqrt{(\varphi(t)+a)(\varphi(t)+b)}-\varphi(t)-G(a,b)>0$,$t\in S$.\label{4} By (\ref{3}) and (\ref{4}), we will
have\[\int_0^1\bigg(\sqrt{(\varphi(t)+a)(\varphi(t)+b)}-\varphi(t)-G(a,b)\bigg)\
dt>0.\]Hence, $ I_{\varphi}>G.$\end{proof}
\begin{prop}
If$\ a\neq b,\ $then\\\\(i)\ $J_{Id}>H,$\\\\(ii)\ neither $\
J_{Id}>G\ $nor$\ J_{Id}<G$,\\\\(iii)\ neither $\ L_{0}<G\ $nor$\
L_{0}>H$,\\\\(iv)\ neither $\ \mathcal{N}_{0}<G\ $nor$\
\mathcal{N}_{0}>H.$\end{prop}
\begin{proof}Let$\ a\neq b\ .$\\(i):$\
J_{Id}>H\Leftrightarrow\frac{(a-b)^2}{2(a+b)}>\frac{(a-b)^2}{4}\ln\big(1+\frac{2}{a+b}\big)\Leftrightarrow\frac{2}{a+b}>\ln(1+\frac{2}{a+b})\checkmark.$\\
(ii) One counter example is : $\\ J_{Id}(0.5,1)<0.6971<0.7071<G(0.5,1).\ \ $On the other hand,$\
J_{Id}(0.5,0.2)>0.31962>0.31623>G(0.5,0.2).$\\(iii) A counter examples would be: $\
L_0(0.1,0.2)>0.14516>0.14143>G(0.1,0.2).\\ $On the other hand, we have:\\$\
L_0(4.1754412,4.175399)-H(4.1754412,4.175399)<-10^{-9}<0.$\\(iv) Here is a counter example: \[\mathcal{N}_0(0.5,0.2)>0.34713>0.31623>G(0.5,0.2)\].\\ On the other hand, 
$\mathcal{N}_0(4.1,4.100000001)-H(4.1,4.100000001)<-10^{-19}<0.$\end{proof}
\begin{thm}Let $\ M\in\mathcal{R}(\mathbb{R}_{+}^2)\ $ be a
mean function and $\ \varphi:[\alpha_0,\beta_0]\to\mathbb{R}_{+} \
$be an integrable function; that is, $\
\varphi\in\mathcal{R}([\alpha_0,\beta_0])$, where$\
\alpha_0,\beta_0\in\mathbb{R}\ $and$\ \alpha_0<\beta_0.\ $Also,$\
\psi:\mathbb{R}_{+}\to\mathbb{R}_{+}\ $is a Lipschitz function
with the constant of 1.\\
Then $\mathcal{P}_{M,\varphi,\psi}:=\mathcal{P}_{M,\varphi,\psi}(a,b)$, defined in the following way:\\
\[\mathcal{P}_{M,\varphi,\psi}(a,b):= A(a,b)-A(\psi(a),\psi(b))+\frac{1}{(\beta_0-\alpha_0)\bar{\varphi}}\int_{\alpha_0}^{\beta_0}M\big(\varphi(t)\psi(a),\varphi(t)\psi(b)\big)\
    dt\]
is a mean of $a$ and $b$.\end{thm}
\begin{proof}
\[\min\{a,b\}=A(a,b)-\frac{|a-b|}{2}\leq A(a,b)-\frac{|\psi(a)-\psi(b)|}{2}=\]\[A(a,b)-A(\psi(a),\psi(b))+\min\{\psi(a),\psi(b)\}\leq
\mathcal{P}_{M,\varphi,\psi}(a,b)\leq...\]\[
A(a,b)-A(\psi(a),\psi(b))+\max\{\psi(a),\psi(b)\}=A(a,b)+\frac{|\psi(a)-\psi(b)|}{2}\]\[\leq
A(a,b)+\frac{|a-b|}{2}=\max\{a,b\}.\]Also, it is obvious that$\
\mathcal{P}_{M,\varphi,\psi}\ $is symmetric.\end{proof}
\begin{remark}
Since$\
\varphi\in\mathcal{R}([\alpha_0,\beta_0])\Leftrightarrow\varphi\circ\eta\in\mathcal{R}([0,1]),\
$\\where$\ \eta(t):=(\beta_0-\alpha_0)t+\alpha_0.\ $So, without
loss of generality, we can assume$\ \alpha_0=0\ $and\\$\
\beta_0=1.$\end{remark}
\begin{remark}If$\ \ \psi(a)=\psi(b),\ \ $then$\ \
\mathcal{P}_{M,\varphi,\psi}(a,b)=A(a,b).$\end{remark}
\begin{remark}If$\
\varphi=c>0\ $($\ c$ is a constant),\ \
then\[\mathcal{P}_{M,c,\psi}=\mathcal{P}_{M,c,\psi}(a,b)={\frac{1}{c}}M(c\psi(a),c\psi(b))-A(\psi(a),\psi(b))+A(a,b).\]\end{remark}
\begin{remark}\[(\mathcal{P}_{M,c,\psi})_{A}=(\mathcal{P}_{M,c,\psi})_{A}(a,b)=-{\frac{1}{c}}M(c\psi(a),c\psi(b))+A(\psi(a),\psi(b))+A(a,b).\]\end{remark}
\begin{remark}
If$\ M\ $is a homogeneous function of order 1; that
is\[M(xz,yz)=zM(x,y),\hskip1cm\forall x,y,z\in\mathbb{R}_+\
,\]then$\ \
\mathcal{P}_{M,\varphi,\psi}(a,b)=A(a,b)-A(\psi(a),\psi(b))+M(\psi(a),\psi(b)).$\end{remark}
\begin{remark}$C,r,g,A,Hn,G\ $and$\ H\ $are homogeneous mean functions
of order 1.\end{remark}
\begin{prop} Let $M,M_1$ and $M_2$ be
mean functions and $\lambda\in\mathbb{R}.$Then\\\\
(i) The following three equations hold:\\
$\mathcal{P}_{A,\varphi,\psi}(a,b)=A(a,b),$\\$\mathcal{P}_{G,\varphi,\psi}(a,b)=A(a,b)-A(\psi(a),\psi(b))+G(\psi(a),\psi(b))=A(a,b)-{\frac{1}{2}}\big(\sqrt{\psi(a)}-\sqrt{\psi(b)}\big)^2$,\\
$(\mathcal{P}_{G,\varphi,\psi})_A(a,b)=A(a,b)+{\frac{1}{2}}
\big(\sqrt{\psi(a)}-\sqrt{\psi(b)}\big)^2,$\\\\(ii)
$M_1>M_2\Rightarrow\mathcal{P}_{M_{1},\varphi,\psi}>\mathcal{P}_{M_{2},\varphi,\psi},\
\ a\neq b,$\\\\In particular, $\
M>(<)A\Rightarrow\mathcal{P}_{M,\varphi,\psi}>(<)A,$\\\\(iii)$\
\mathcal{P}_{\lambda M_1+(1-\lambda)
M_2,\varphi,\psi}=\lambda\mathcal{P}_{M_1,\varphi,\psi}+(1-\lambda)\mathcal{P}_{M_2,\varphi,\psi},\
$if$\ \lambda M_{1}+(1-\lambda) M_{2}\ $is a mean
function.\\\\
In particular,$\
\mathcal{P}_{M_{A},\varphi,\psi}=(\mathcal{P}_{M,\varphi,\psi})_A.$\end{prop}
\begin{proof}It is straightforward.\end{proof}
Here are some examples where the above proposition is used: \begin{exm}Let$\ \ M=G\ $and $\
\check{\psi}(t)=\frac{t-\sin t}{2},\
$then\[\mathcal{P}_{G,\varphi,\check{\psi}}=\mathcal{N}_0.\]\end{exm}
\begin{exm}
Let$\ \ M=G\ $and$\ \tilde{\psi}(t)=\ln(t^2+1),\
$then\[\mathcal{P}_{G,\varphi,\tilde{\psi}}=L_0.\] \end{exm}
\begin{exm}
Let$\ \ M=H,\ \ \varphi(t)=Id(t)=t,\
$then\[\mathcal{P}_{H,Id,\psi}(a,b)=A(a,b)-A(\psi(a),\psi(b))+H(\psi(a),\psi(b))=A(a,b)-\frac{\big(\psi(a)-\psi(b)\big)^2}{2\big(\psi(a)+\psi(b)\big)}
.\]In particular,\[\mathcal{P}_{H,Id,Id}(a,b)=H(a,b).\]We can
see\[\mathcal{P}_{H,Id,\psi}(a,b)<\mathcal{S}_{H,Id,\psi}(a,b),\ \
\ \ \ \ a\neq b.\]
\end{exm}

Let$\ M\ $be a mean function. We define\[
\hat{S}_{M}:=\hat{S}_{M}(a,b):=\int_0^{\frac{\pi}{2}}M(a\sin\theta,b\cos\theta)\
d\theta.\]
\\ We can easily see$\ \
\hat{S}_{M}(a,b)=\hat{S}_{M}(b,a).\ $Also,\[
\hat{S}_{M}(a,b)\leq\int_0^{\frac{\pi}{2}}\big(\frac{a\sin\theta+b\cos\theta}{2}+\frac{|a\sin\theta-b\cos\theta|}{2}\big)\
d\theta=\]\[A(a,b)+{1\over2}\int_0^{\tan^{-1}{b\over
a}}(b\cos\theta-a\sin\theta)\
d\theta+{1\over2}\int_{\tan^{-1}{b\over
a}}^{\frac{\pi}{2}}(-b\cos\theta+a\sin\theta)\
d\theta=\]\[\sqrt{a^2+b^2}.\label{8}\]
\[
\hat{S}_{M}(a,b)\geq\int_0^{\frac{\pi}{2}}\big(\frac{a\sin\theta+b\cos\theta}{2}-\frac{|a\sin\theta-b\cos\theta|}{2}\big)\
d\theta=\]\[A(a,b)-{1\over2}\int_0^{\tan^{-1}{b\over
a}}(b\cos\theta-a\sin\theta)\
d\theta-{1\over2}\int_{\tan^{-1}{b\over
a}}^{\frac{\pi}{2}}(-b\cos\theta+a\sin\theta)\
d\theta=\]\[2A(a,b)-\sqrt{a^2+b^2}.\label{9}\] Thus, by (\ref{8}) and
(\ref{9}), we have\be
2A(a,b)-\sqrt{a^2+b^2}\leq\hat{S}_M\leq\sqrt{a^2+b^2}.\label{10}\ee
From (\ref{10}), if we take$\ \xi:=\xi(a,b)>0\ $and$\
\zeta:=\zeta(a,b),\ $such
that\be\min\{a,b\}\leq\frac{(a+b)\xi}{\sqrt{a^2+b^2}}-\xi+\zeta\leq\frac{\xi}{\sqrt{a^2+b^2}}\hat{S}_M+\zeta\leq\xi+\zeta\leq\max\{a,b\},\label{11}\ee
then we will
have\be0<\xi\leq\frac{|a-b|\sqrt{a^2+b^2}}{2\sqrt{a^2+b^2}-(a+b)}\label{12}\ee
and\be\min\{a,b\}+(1-\frac{a+b}{\sqrt{a^2+b^2}})\xi\leq\zeta\leq\max\{a,b\}-\xi\label{13}.\ee
By (\ref{11}), (\ref{12}) and (\ref{13}), we infer\be
S_{M,\xi,\zeta}:=S_{M,\xi,\zeta}(a,b):=\frac{\xi(a,b)}{\sqrt{a^2+b^2}}\hat{S}_M(a,b)+\zeta(a,b)\label{14}\ee
is a mean of $a$ and $b$.\\For example, if we take$\ \xi:=|a-b|\
$and$\
\zeta:={1\over2}\big(\min\{a,b\}+(1-\frac{a+b}{\sqrt{a^2+b^2}})|a-b|+\\\max\{a,b\}-|a-b|\big)=A(a,b)-\frac{|a^2-b^2|}{2\sqrt{a^2+b^2}},\
$then from (\ref{14})\be
S_{M}:=S_{M}(a,b):=A(a,b)-\frac{|a^2-b^2|}{2\sqrt{a^2+b^2}}
+\frac{|a-b|}{\sqrt{a^2+b^2}}\int_0^{\frac{\pi}{2}}M(a\sin\theta,b\cos\theta)\
d\theta\label{15}\ee is a mean of $a$ and $b$. Thus, we will have
the following theorem\\\\{\bf Theorem5}{\it\ \ Let $M$ is a mean
function. Then$\ S_M\ $which is defined by (\ref{15}), is a mean
function.}\\
\begin{prop}\label{11} \ \ Let$\ M,M_1\ $and$\ M_2\
$be mean functions and$\ \lambda\in\mathbb{R}.\ $Then\\\\(i)\
$S_{\lambda M_1+(1-\lambda)M_2}=\lambda
S_{M_1}+(1-\lambda)S_{M_2},\ $if$\ \lambda M_{1}+(1-\lambda)
M_{2}\ $is a mean function.\\\\Specially,\ $S_{M_{A}}=(S_{M})_A,$
\\\\(ii)\
$M_1>M_2\Rightarrow S_{M_1}>S_{M_2},\ \ a\neq b,$\\\\specially,
$M>(<)A\Rightarrow S_M>(<)A.$\end{prop}
\begin{proof} is
straightforward.\end{proof}

\begin{exm}\[\hskip1cm S_A=A,\hskip.7cm S_G(a,b)=A(a,b)+
\frac{G(a,b)|a-b|\Gamma^2({3\over4})}{\sqrt{\pi(a^2+b^2)}}
-\frac{|a^2-b^2|}{2\sqrt{a^2+b^2}},\ \ \ \
\]\[\Gamma(\tfrac{3}{4})\approx1.225416702.\]\end{exm}
\begin{exm}\[
S_H(a,b)=A(a,b)-\frac{|a^2-b^2|(a^2+b^2-4ab)}{2(a^2+b^2)^{3\over2}}-...\]\[\frac{4a^2b^2|a-b|}{(a^2+b^2)^{2}}
\ln\frac{a+b+\sqrt{a^2+b^2}}{\sqrt{2ab}}.\]
\end{exm}
\begin{exm}
By proposition \ref{11} (i)
\[\hskip2cm S_g=A(a,b)+\frac{|a^2-b^2|(a^2+b^2-4ab)}{6(a^2+b^2)^{3\over2}}+...\]\[\frac{4a^2b^2|a-b|}{3(a^2+b^2)^{2}}
\ln\frac{a+b+\sqrt{a^2+b^2}}{\sqrt{2ab}}\]\end{exm}and
\begin{exm}\[\hskip2cm S_C=A(a,b)+\frac{|a^2-b^2|(a^2+b^2-4ab)}{2(a^2+b^2)^{3\over2}}+...\]\[\frac{4a^2b^2|a-b|}{(a^2+b^2)^{2}}
\ln\frac{a+b+\sqrt{a^2+b^2}}{\sqrt{2ab}}.\]
By proposition\ref{11} (ii),
for $a\neq
b$\\\[\bigg(\frac{\sqrt{2(a^2+b^2)}}{|a-b|}\bigg)S_C-\frac{a+b}{\sqrt{2}}\bigg(\frac{\sqrt{a^2+b^2}}{|a-b|}-1\bigg)
\]\[>\int_0^{\frac{\pi}{2}}\sqrt{a^2\sin^2\theta+b^2\cos^2\theta}\
d\theta>\]\[\bigg(\frac{\sqrt{2(a^2+b^2)}}{|a-b|}\bigg)S_g-\frac{a+b}{\sqrt{2}}\bigg(\frac{\sqrt{a^2+b^2}}{|a-b|}-1\bigg)\]and\\
\[S_C(a,b)>S_g(a,b)>A(a,b)>S_G(a,b)>S_H(a,b),\ \ \ \ a\neq b.\]Specially,$\ A(3,4)>S_G(3,4)>S_H(3,4),\ $which we infer
\[\frac{7}{12}>\frac{\Gamma^2(\tfrac{3}{4})}{\sqrt{3\pi}}>\frac{140-48\ln\!6}{125}.\] \end{exm}
\begin{thm} \ Let$\ M_1,
M_2\in\mathcal{R}(\mathbb{R}_{+}^2)\ $be mean functions.\ Then\\\[\mathcal{T}_{M_{1},M_{2}}:=\mathcal{T}_{M_{1},M_{2}}(a,b):=\left\{%
\begin{array}{ll}
    \frac{1}{b-a}\int_a^bM_1\big(M_2(a,b),x\big)\ dx, & \hbox{$a\neq b,$} \\\\
    a, & \hbox{$a=b$} \\
\end{array}%
\right.     \]is a mean of$\ a\ $and$\ b$.\end{thm}
\begin{proof}
Let$\ b>a\ $and$\ x\in[a,b].\ $We have\[
\frac{1}{b-a}\int_a^bM_1\big(M_2(a,b),x\big)\
dx\leq\frac{1}{b-a}\int_a^b\max\{M_2(a,b),x\}\
dx=\]\[\frac{1}{b-a}\int_a^{M_2(a,b)}M_2(a,b)\
dx+\frac{1}{b-a}\int_{M_2(a,b)}^bx\
dx=...\]\[\frac{1}{2(b-a)}\big(b^2+M_2^2(a,b)-2aM_2(a,b)\big).\label{p1}\].If
we take$\ m_1(t):=b^2+t^2-2at,\ \ t\in[a,b],\ $then$\ m_1\ $will
be increasing on$\ [a,b].\ $So, $m_1(t)\leq m_1(b)=2b(b-a),\
$for$\ t\in[a,b].$ Hence, from (\ref{p1}), we will
get\[\frac{1}{b-a}\int_a^bM_1\big(M_2(a,b),x\big)\ dx\leq
b.\]Similarly,\[ \frac{1}{b-a}\int_a^bM_1\big(M_2(a,b),x\big)\
dx\geq\frac{1}{b-a}\int_a^b\min\{M_2(a,b),x\}\
dx=\]\[\frac{1}{b-a}\int_a^{M_2(a,b)}x\
dx+\frac{1}{b-a}\int_{M_2(a,b)}^bM_2(a,b)\
dx=...\]\[\frac{1}{2(b-a)}\big(-a^2-M_2^2(a,b)+2bM_2(a,b)\big).\label{p2}.\]
If
we take$\ m_2(t):=-a^2-t^2+2bt,\ \ t\in[a,b],\ $then$\ m_2\ $will
be increasing on$\ [a,b].\ $ So, $m_2(t)\geq m_2(a)=2a(b-a),\
$for$\ t\in[a,b].$ Therefore, from (\ref{p2}), we will
get\[\frac{1}{b-a}\int_a^bM_1\big(M_2(a,b),x\big)\ dx\geq
a.\]Also, it is obvious that$\ \mathcal{T}_{M_{1},M_{2}}\ $ is
symmetric.\end{proof}
\begin{prop}
\label{12}\ Let$\
M_1,M_1^{'},M_2\ $and$\ M_2^{'}\ $be mean functions and
$\lambda\in\mathbb{R}$.\ Then\\\\(i)$\
M_1>M_1^{'}\Rightarrow\mathcal{T}_{M_{1},M_{2}}(a,b)>\mathcal{T}_{M_{1}^{'},M_{2}}(a,b),\
\ a\neq b,$\\\\(ii) If$\ M_1\ $is strictly increasing and$\
M_2>M_2^{'},\ \ $then \[
\mathcal{T}_{M_{1},M_{2}}(a,b)>\mathcal{T}_{M_{1},M_{2}^{'}}(a,b),\
\ a\neq b,\]\\\\(iii)\ $\mathcal{T}_{\lambda M_{1}+(1-\lambda)
M_{1}^{'},M_{2}}=\lambda\mathcal{T}_{M_{1},M_{2}}+(1-\lambda)\mathcal{T}_{M_{1}^{'},M_{2}},\
$\\if$\ \lambda M_{1}+(1-\lambda) M_1^{'}\ $is a mean
function.\\Specially,
$\mathcal{T}_{M_{1A},M_{2}}=2\mathcal{T}_{A,M_{2}}-\mathcal{T}_{M_{1},M_{2}}$.\end{prop}
\begin{proof}\ is straightforward.\end{proof}
{\bf Some Examples} \ Let$\
M\ $be a mean
function.\[(1)\hskip6cm\mathcal{T}_{A,M}(a,b)=A\big(M(a,b),A(a,b)\big).\]
\[(2)\hskip5.3cm\mathcal{T}_{G,M}(a,b)=G\big(M(a,b),g^2(\sqrt{a},\sqrt{b})\big).\]
\[(3)\hskip1.3cm\mathcal{T}_{H,M}(a,b)=2M(a,b)\bigg(1-\frac{M(a,b)}{b-a}\ln\frac{b+M(a,b)}{a+M(a,b)}\bigg),\ \ a\neq
b.\]
\[(4)\hskip1.5cm\mathcal{T}_{r,M}(a,b)=\frac{1}{2\sqrt{2}}\bigg(\frac{(a+b)\big(a^2+b^2+M^2(a,b)\big)}{\big(a\sqrt{a^2+M^2(a,b)}+b\sqrt{b^2+M^2(a,b)}\big)}+\]\[
\frac{M^2(a,b)}{(b-a)}\ln\frac{b+\sqrt{b^2+M^2(a,b)}}{a+\sqrt{a^2+M^2(a,b)}}\bigg),\
\ a\neq b.\]
 By proposition \ref{12}, we will get
\[\mathcal{T}_{r,M}(a,b)>\mathcal{T}_{A,M}(a,b)>\mathcal{T}_{G,M}(a,b)>\mathcal{T}_{H,M}(a,b),\ \ \
\ a\neq b.\]Also, by proposition\ref{12}(ii) and note12, we will
have\[\mathcal{T}_{A,A}(a,b)>\mathcal{T}_{A,G}(a,b)>\mathcal{T}_{A,H}(a,b),\
\ \ \ a\neq
b,\]\[\mathcal{T}_{G,A}(a,b)>\mathcal{T}_{G,G}(a,b)>\mathcal{T}_{G,H}(a,b),\
\ \ \ a\neq
b\]and\[\mathcal{T}_{H,A}(a,b)>\mathcal{T}_{H,G}(a,b)>\mathcal{T}_{H,H}(a,b),\
\ \ \ a\neq b.\]Besides, by proposition\ref{12}(iii), we will have
\[\mathcal{T}_{H_n,M}=\tfrac{2}{3}\mathcal{T}_{A,M}+\tfrac{1}{3}\mathcal{T}_{G,M},\]
\[\mathcal{T}_{g,M}=\tfrac{4}{3}\mathcal{T}_{A,M}-\tfrac{1}{3}\mathcal{T}_{H,M}\]and
\[\mathcal{T}_{(M_1)_{A},M_2}=2\mathcal{T}_{A,M_2}-\mathcal{T}_{M_1,M_2},\]if$\ M,M_1\
$and$\ M_2\ $are mean functions.

\bibliographystyle{amsplain}

\end{document}


@author:
@affiliation:
@title:
@language: English
@pages:
@classification1:
@classification2:
@keywords:
@abstract:
@filename:
@EOI